\documentclass[letterpaper]{article} 
\usepackage{aaai24}  
\usepackage{times}  
\usepackage{helvet}  
\usepackage{courier}  
\usepackage[hyphens]{url}  
\usepackage{graphicx} 
\usepackage{natbib}  
\usepackage{caption} 
\usepackage{algorithm}

\newcommand{\ModelName}{E3D2}
\usepackage[noend]{algpseudocode}
\usepackage{bm}

\usepackage{array}
\usepackage{multirow}
\usepackage{color}
\usepackage{tabularx}
\usepackage{amsmath}
\usepackage{amsfonts}
\usepackage{amssymb,amsthm}
\usepackage{mathtools}
\newtheorem{theorem}{Theorem}
\usepackage{mathrsfs}
\usepackage{booktabs}
\usepackage{euscript}
\usepackage{graphicx}
\usepackage[T1]{fontenc}
\DeclareMathOperator*{\argmax}{arg\,max}

\newcolumntype{C}{>{\centering\arraybackslash}X}

\usepackage{newfloat}
\usepackage{listings}
\DeclareCaptionStyle{ruled}{labelfont=normalfont,labelsep=colon,strut=off} 
\floatstyle{ruled}
\newfloat{listing}{tb}{lst}{}
\floatname{listing}{Listing}
\title{Explore 3D Dance Generation via Reward Model from Automatically-Ranked Demonstrations}
\author{
    Zilin Wang\textsuperscript{\rm 1\equalcontrib},
    Haolin Zhuang\textsuperscript{\rm 1\equalcontrib},
    Lu Li\textsuperscript{\rm 1\equalcontrib},
    Yinmin Zhang\textsuperscript{\rm 2},
    Junjie Zhong\textsuperscript{\rm 3},\\
    Jun Chen\textsuperscript{\rm 1},
    Yu Yang\textsuperscript{\rm 1},
    Boshi Tang\textsuperscript{\rm 1},
    Zhiyong Wu\textsuperscript{\rm 1}\thanks{Corresponding author.}
}
\affiliations{
    \textsuperscript{\rm 1}Shenzhen International Graduate School, Tsinghua University\\
    \textsuperscript{\rm 2}University of Sydney,
    \textsuperscript{\rm 3}Waseda University\\

    \{wangzl21, zhuanghl21, lilu21\}@mails.tsinghua.edu.cn, 
     yinmin.zhang@sydney.edu.au, junjiezhong@ruri.waseda.jp\\
    \{y-chen21, yy20, tbs22\}@mails.tsinghua.edu.cn,
    zywu@sz.tsinghua.edu.cn

}

\begin{document}

\maketitle

\begin{figure*}[h]
\centering
  \includegraphics[width=0.85\textwidth]{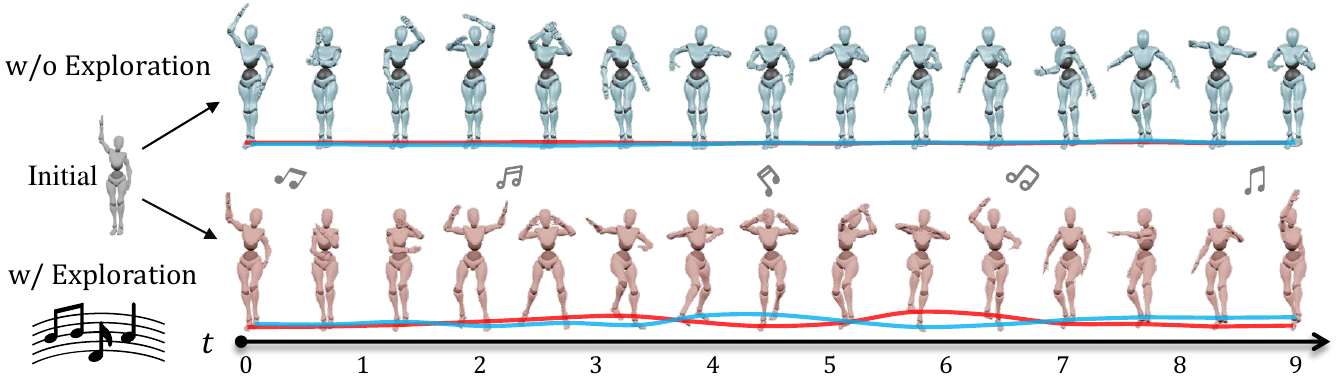}
  \vspace{-5pt}
  \caption{
    Visualizations. Red and blue lines represent right and left leg movements, respectively. 
    \textit{Top}: Dance examples generated by the policy lack exploration, exhibiting limited leg movements and diversity. 
    \textit{Bottom}: Dance examples generated by the policy reinforced via exploration align with human preferences, showcasing complex leg movements and increased diversity.
  }
  \label{fig:specimen}
\vspace{-10pt}
\end{figure*}

\begin{abstract}

This paper presents an Exploratory 3D Dance generation framework, \ModelName, designed to address the exploration capability deficiency in existing music-conditioned 3D dance generation models. 
Current models often generate monotonous and simplistic dance sequences that misalign with human preferences because they lack exploration capabilities.
The \ModelName~framework involves a reward model trained from automatically-ranked dance demonstrations, which then guides the reinforcement learning process. 
This approach encourages the agent to explore and generate high quality and diverse dance movement sequences. 
The soundness of the reward model is both theoretically and experimentally validated. 
Empirical experiments demonstrate the effectiveness of E3D2 on the AIST++ dataset.
Project Page: https://sites.google.com/view/e3d2.


\end{abstract}


\section{Introduction}

Music-conditioned 3D dance generation is an emerging field that combines the art of dance and the science of machine learning, fostering a novel and creative fusion.
By utilizing music as a guiding condition, dance generative models create dance poses synchronized with the melody and rhythm of the music.
Several studies \cite{huangdance, huang2022genre, li2021ai, siyao2022bailando} utilize generative networks to auto-regressively generate dance sequences in supervised learning, with music as the condition and human dance poses as the supervisory signal.
These approaches are capable of producing complete dance movements, as significant advancements in the field of dance generation.

Nevertheless, we observe that supervised learning approaches often exhibit the following three shortcomings:
(1) Weak generalization for unseen music, which affects diversity and quality,
(2) Fragility of auto-regressive models, which are prone to severe compounding rollout errors, particularly when data is scarce, leading to the potential collapse of the dance sequence, and
(3) Misalignment between generated dances and human preferences, which stems from the excessive focus on mimicking human movements without considering human preferences (\textit{e.g.}, movement difficulty and aesthetic appeal). 
Inspired by the learning process of human dancers, novice dancers not only require the mechanical imitation of movements from dance experts but also continuous practice and exploration to develop their skills. 
Furthermore, receiving feedback from experts plays a crucial role in reinforcing their movements, ultimately helping them become proficient dancers.

In this work, we argue that the aforementioned three issues arise from the lack of \textbf{exploration} capacity in current dance generation models.
We expect trained dance agents to explore various movements within the dance space while receiving accurate signals indicating which movements are desirable, thereby increasing the probability of generating such movements.
Based on this assumption, we propose the \textbf{Exploratory 3D Dance generation framework, \ModelName,} to address the issue of exploration.
To achieve this, we model the music-conditioned dance generation task as a Markov Decision Process (MDP) and employ Reinforcement Learning (RL) to endow the dance agent with the ability to explore.
For the reward signal, we utilize Inverse Reinforcement Learning (IRL) to train a Reward Model (RM) from automatically ranked dance demonstrations, which guides the exploration and exploitation of the dance agent.
As shown in Figure \ref{fig:diagram}, we firstly use Behavior Cloning \cite{michie1990cognitive} to train an initial dance generation policy, allowing the agent to learn basic dance movements.
Then, we inject increasing levels of noise into multiple cloned initial dance generation policies to acquire multiple policies with decreasing performance, generating different quality dance demonstrations.
Next, we train a reward model with these automatically ranked dance demonstrations.
Finally, under the guidance of the learned reward model, we encourage the dance agent to explore using reinforcement learning, ultimately obtaining the optimal dance generation policy.

Our design enables the dance policy to address the above issues through exploration:
(1) To tackle the limited diversity of generated dances, 
reinforcement learning encourages the dance agent to efficiently explore a broader range of state-action pairs, where new movements emerge naturally, either by combining the sub-actions from various human dances or by creating entirely novel dances.
This results in more diverse dance movements, as shown in Figure \ref{fig:specimen}.
Moreover, the consistent distribution of music and dance movement in the environment ensures the stability of dynamics, such as transition probability, allowing the learned reward to generalize. This theoretically guarantees the generalization of the learned policies, guided by the learned reward, for both seen and unseen environments.
(2) Regarding the fragility of auto-regressive models under supervised learning, which suffers from severe compounding rollout errors of single-step decisions with respect to long planning horizons, our proposed method is optimized by sequence-based reward with trial and error.
Through sequence-based exploration and exploitation, our proposed method focuses on the generated dance trajectories rather than its single-step error, avoiding the compounding errors \cite{asadi2019combating, janner2021offline, janner2022planning}.
(3) To address the misalignment between the policy and human preferences, our proposed reward model is able to distinguish the differences between attractive and ordinary dances due to the assumption that the dance generation with a higher noise level aligns less with human preferences.
During the exploration and exploitation, human preference is incorporated into the dance policy through the guidance of the reward model.


Empirical experiments on the AIST++ dataset \cite{li2021ai} demonstrate that the proposed \ModelName~ outperforms the behavior cloning (pure supervised) method across multiple metrics. 
Moreover, we perform an in-depth analysis and provide a theoretical proof (in Appendices) of the reward model.
The contributions of this article are three-fold:
\begin{itemize}
    \item We illuminated three issues, weak generalization, fragility, and misalignment, in existing supervised dance generation methods attributable to a lack of exploration capability.
    \item To address the deficiency of exploration, we propose an Exploratory 3D Dance generation framework, \ModelName, which encourages dance agents to explore by introducing the inverse reinforcement learning method with a learned reward model that reflects human preference.
    \item Empirical experiments demonstrate the effectiveness and generalization performance of our reward model and \ModelName~over supervised models.
\end{itemize}

\begin{figure*}[!htpb]
\centering
    \includegraphics[width=0.85\linewidth]{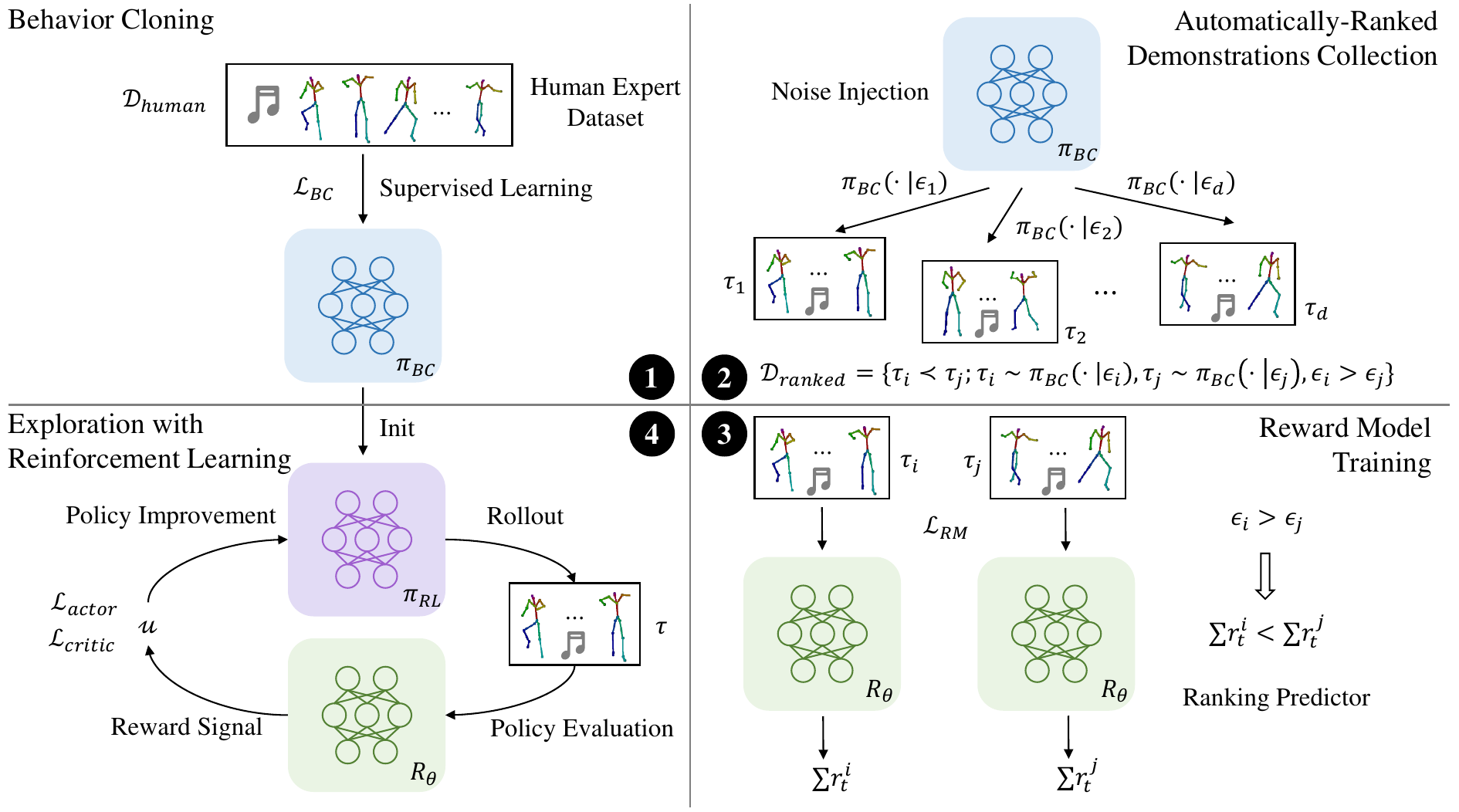}
    \vspace{-5pt}
\caption{Diagram of our \ModelName: (1) An initial policy $\pi_{BC}$ is distilled from the human expert dataset through behavior cloning. 
    (2) Automatically ranked dance demonstrations are collected by $\pi_{BC}$ with different levels of noise.
    (3) A reward model $R_\theta$ is trained from these automatically ranked demonstrations to rank the quality of dance trajectories.
    (4) A reinforcement learning policy $\pi_{RL}$ is initialized with $\pi_{BC}$ and optimized to obtain the optimal dance policy, guided by the reward model $R_\theta$.}
\vspace{-10pt}
\label{fig:diagram}
\end{figure*}
\section{Related Works}

\subsection{Music-conditioned Dance Generation}
Music-conditioned dance generation is a cross-modal task involving auditory and visual integration.
Existing methods for music-conditioned dance generation can be broadly classified into two categories: retrieval-based methods and direct generation methods.
Retrieval-based methods \cite{ofli2012learn2dance, fan2011example,fukayama2015music, lee2013music,ye2020choreonet, chen2021choreomaster,au2022choreograph} divide dances into fixed length units and choreograph by concatenating these units according to the melody of the music.
Unfortunately, 
the fixed length and Beats Per Minute (BPM) of the segmented dance units imposed significant restrictions on the rhythm of the music used to drive the dance.
To tackle these issues, direct generation methods \cite{alemi2017groovenet,tang2018anidance,ahn2020generative,huangdance, huang2022genre,zhuang2022music,valle2021transflower,wang2022groupdancer,gao2022pc,li2022danceformer,li2020learning,tseng2023edge} have been proposed which generate dance motion from scratch. 
These methods are trained in a supervised learning fashion, with music as the conditioning input and real human dance as the supervisory signal.
In this work, we focus on exploratory capabilities during training to improve the quality and diversity of the generated dance sequences.

\subsection{Preference-based Inverse Reinforcement Learning}

The goal of Preference-based Inverse Reinforcement Learning (PIRL) \cite{cheng2011preference, sugiyama2012preference, wirth2017survey, christiano2017deep} is to learn a reward function from expert preferences. 
Compared with learning the reward model directly from expert behaviors through conventional IRL methods \cite{russell1998learning, ng2000algorithms, abbeel2004apprenticeship}, \textit{e.g.}, Adversarial Inverse Reinforcement Learning (AIRL) \cite{fu2018learning}, PIRL have been effectively applied in many high-dimensional state spaces \cite{brown2020better,ibarz2018reward} where AIRL may not work effectively \cite{tucker2018inverse}.
Besides, PIRL could also serve as a way to introduce the human feedback (RLHF) \cite{christiano2017deep, warnell2018deep,macglashan2017interactive}, make the model better aligned with human preferences \cite{stiennon2020learning, wu2021recursively, nakano2021webgpt, ganguli2022red, glaese2022improving}.
To address the issue of sub-optimal demonstrations, T-REX \cite{brown2019extrapolating} trained a reward model conditioned on states with expert-provided ranking information and then trained an agent that surpasses the sub-optimal demonstrator using the reward model.
Based on T-REX, D-REX \cite{brown2020better} proposed a generation method of automatically ranked demonstrations by injecting different levels of noise into the behavior cloning policy.
D-REX is highly relevant to the demonstration collection of \ModelName.
However, our main focus is not so much that we proposed a novel PIRL algorithms, or our successful adoptation of D-REX in dance generation, but rather our methods solve exploration capability deficiency plaguing existing music-conditioned 3D dance generation models, that were previously unaddressed and holds significant importance.

\section{Preliminary}\label{sec:preliminary}
Given a music-driven dance dataset $\mathcal{D}=\{(\bm{m}^i,\bm{p}^i)\}_{i=1}^N$ consisting of $N$ sequence pairs, where $\bm{m}^i\in \mathcal{M}$ is a music feature sequence, and $\bm{p}^i\in \mathcal{P}$ is the corresponding dance sequence, $\mathcal{M}$ and $\mathcal{P}$ represent the music feature space and the dance motion space, respectively.
We treat music-conditioned dance generation as a sequential decision problem \cite{sutton2018reinforcement} and model it as a Markov Decision Process (MDP) $(\mathcal{S},\mathcal{A},R,P,\gamma,T)$, where $\mathcal{S},\mathcal{A}$ represent state and action spaces, $T$ and $R$ represent the termination of the episode and reward function, and $\gamma \in (0, 1)$ represents the discount factor.
We identify two entities, the \textit{environment} and the \textit{agent}, where the environment is determined by MDP the agent is determined by the policy $\pi$.
To sufficiently consider the consistency of the dance generation sequence, we instantiate the MDP by extending the state with history information.
At the beginning of each episode, $t=0$, the dance \textit{agent} receives the initial state $s_0=\{m_{init},p_{init},m_0\}\in \mathcal{S}$, which is randomly sampled from the dataset by the \textit{environment}, where $\mathcal{S}$ is the state space with $s_t\in \mathcal{P}^{t+1}\times\mathcal{M}^{t+2}$ and $m_{init}$ and $p_{init}$ are the initial music feature and dance pose, respectively. 
Then, the agent generates an action $a_0\sim\pi(\cdot|s_0) \in \mathcal{A}$ according to the policy $\pi$, where the action space $\mathcal{A}=\mathcal{P}$ and thus $a_t=\hat{p}_t$. 
Then, the environment receives the action $\hat{p}_0$ and obtains the next state $s_1=\{m_{init},p_{init},m_0,\hat{p}_0,m_1\}$ using the deterministic state transition function $P: \mathcal{S}\times \mathcal{A}\rightarrow \Delta(\mathcal{S})=s_t.\operatorname{extend}(\{\hat{p}_t,m_{t+1}\})$.
After that, the reward $r_t$ of taking action $a_t$ at state $s_t$ is obtained from the reward function $R:\mathcal{S}\times \mathcal{A}\rightarrow \mathbb{R}$. 
The agent then continues to make decisions based on $s_1$ to determine the next action $a_1$. 
This process continues until the termination of the episode $T$, where the termination state $s_{T-1}=\{m_{init},p_{init},m_0,\hat{p}_0,\cdots,m_{T-1},\hat{p}_{T-1}\}$ is reached (there is no subsequent music feature $m_{T}$ in the termination state) and target dance sequence $\{p_{init},\hat{p}_0,\cdots,\hat{p}_{T-1}\}$ is obtained. 
The objective of our learning algorithm is to train a dance agent with optimal policy $\pi^*(a|s)$ to maximize the expected discounted return $J(\pi^*)=E_{\tau\sim\pi^*}[\sum_{t=0}^{T-1}\gamma^tr(s_t,a_t)]$, where $\tau=\{m_{init}, p_{init}, m_0,\hat{p}_0,m_1,\hat{p}_1,\cdots,m_{T-1},\hat{p}_{T-1}\}$ and $r(s_t, a_t)=R(s_t,a_t)$ means the reward value, abbreviated as $r_t$.
To simplify the notation, in the following parts, unless otherwise specified, $p_t=\hat{p}_t$.
\section{Methodology}
Our framework comprises four main steps, as illustrated in Figure \ref{fig:diagram}.
The subsequent sections provide a comprehensive depiction of each component.

\subsection{Behavior Cloning}\label{sec:bc}
To supply demonstration data for reward model training and establish the initial skill required for efficient exploration, we learn an initial policy, $\pi_{BC}$, from the human expert dataset $\mathcal{D}_{human}$ in a supervised learning manner. 
Specifically, we follow \cite{siyao2022bailando} including network architecture (\textit{i.e.}, transformer), objective (\textit{i.e.}, cross-entropy loss $\mathcal{L}_{BC}$), and action space discretization (\textit{i.e.}, VQ-VAE).

\subsection{Automatically-Ranked Demonstrations Collection}\label{sec:collect}

In this section, we will describe how to use $\pi_{BC}$ to collect automatically ranked demonstrations. Specifically, we obtain policies with performance between the behavior cloning policy $\pi_{BC}$ and a completely random policy by injecting noise of different levels into the pretrained behavior cloning policy, similar to D-REX \cite{brown2020better}.
Empirically (in Discussion section), we show that given a noise schedule $\mathcal{E}=(\epsilon_1,\epsilon_2,\cdots,\epsilon_d)$ where the $\epsilon_i, i \in \{1.\cdots,d\}$ means the noise range in $[0, 1]$ and are ordered as $\epsilon_1 > \epsilon_2 > \cdots > \epsilon_d$. Intuitively, the dance agent's performance $J(\cdot)$ is likely to have the following ranking: $J(\pi_{BC}(\cdot|\epsilon_1)) < J(\pi_{BC}(\cdot|\epsilon_2)) < \cdots < J(\pi_{BC}(\cdot|\epsilon_d))$.

In practice, we collect demonstrations by using the $\epsilon$-greedy strategy to inject noise into the policy. 
That is, at each decision-making step, the agent has a probability of $\epsilon$ to uniformly sample an action $a$ from the action space $\mathcal{A}$, and a probability of $1-\epsilon$ to decide on the action $a$ based on its learned policy $\pi_{BC}$. 
For each noise $\epsilon_i$, $K$ dance trajectories are generated to construct the dataset for training the reward model. 
Finally, the dataset contains $d\times K$ trajectories with the following ranking relationship:
\begin{equation}\label{eq:dranked}
    D_{ranked}=\{\tau_i\prec\tau_j;\tau_i\sim \pi_{BC}(\cdot|\epsilon_i),\tau_j\sim\pi_{BC}(\cdot|\epsilon_j),\epsilon_i>\epsilon_j\},
\end{equation}
where $\tau_i\prec\tau_j$ means $\tau_i$ is worse than $\tau_j$.

\subsection{Reward Model}\label{sec:rewardmodel}
In this section, given the automatically-ranked demonstrations dataset $D_{ranked}$, we will discuss the network architecture and training method of the reward model.
Intuitively, as a discriminative model, the reward model has greater potential for extrapolation tasks compared to auto-regressive models.

\begin{figure}[t]
	\centering
    \includegraphics[width=\linewidth]{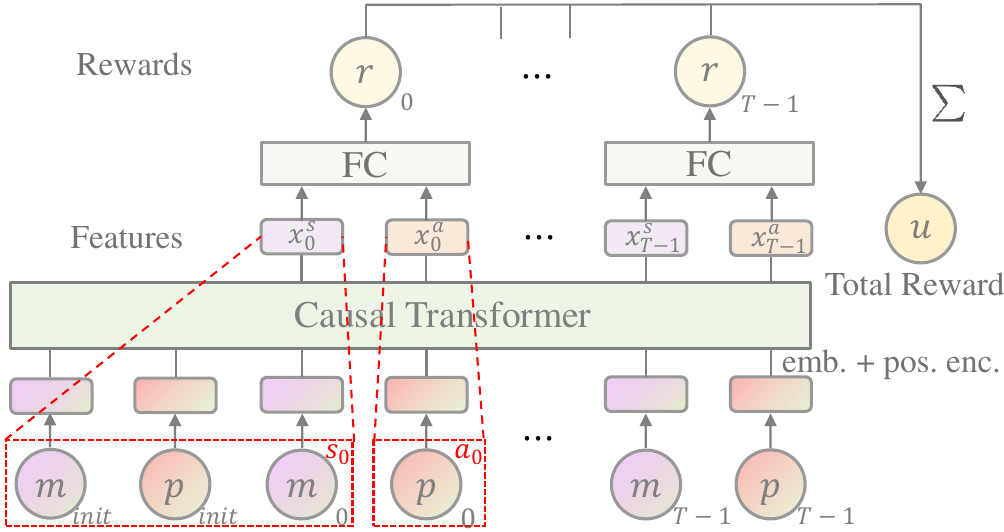}
    \vspace{-5pt}
	\caption{Overview of our reward model. 
 Tokens are generated by combining music and pose embeddings with position encoding, followed by a Causal Transformer to extract features.
The extracted features are then fed into a fully connected layer to predict the total reward.}
\vspace{-10pt}
	\label{fig:reward_model}
\end{figure}


\subsubsection{Network Architecture}
An overview of the reward model is shown in Figure \ref{fig:reward_model}.
Given a dance trajectory $\tau=\{m_{init}, p_{init}, m_0,p_0,m_1,p_1,\cdots,m_{T-1},p_{T-1}\}$ of $T$ timesteps generated by the interaction between the agent and the environment, which contains two modalities, music and dance poses, with a total length of $2(T+1)$.
We interleave the two modalities in the trajectory to ensure compatibility with the standard causal attention mechanism and feed them into the reward model $R_\theta$.
Then, we apply a linear layer for each modality to map the raw inputs into an embedding space, added by a learned timestep embedding, which is shared by different modalities embedding similar to Decision Transformer \cite{chen2021decision}. 
Subsequently, these tokens will be fed into a causal transformer with multiple layers of masked multi-head attention to produce output features with equal length $\{x^s_t,x^a_t\}_{t=0}^{T-1}$, where $x^s_t$ represents the feature of state $s_t$, and $x_t^a$ represents the feature of action $a_t$.
Then, the corresponding state and action features will be fed into a fully connected layer to generate the reward $r_t$ that can be obtained by taking action $a_t$ under the current state $s_t$. 
By applying the reward function $R_\theta(s_t,a_t)$, we obtain the total reward $u=\sum_tr_t$ for the entire sequence.

\subsubsection{Training the reward model}
For training, we first sample a pair of trajectories ${\tau_i,\tau_j}$ of different quality from the automatically ranked demonstration dataset $D_{ranked}=\{\tau_1,\cdots,\tau_m\}$, where $\tau_i\sim \pi_{BC}(\cdot|\epsilon_i),\tau_j\sim\pi_{BC}(\cdot|\epsilon_j)$ and $\epsilon_i\neq\epsilon_j$.
Next, we obtain the quantitative metric of each trajectory through the reward model, \textit{i.e.}, the total reward $u_l=\sum_{s_t,a_t \in \tau_l}R_\theta(s_t,a_t), l\in\{i,j\}$.
Here we use the total reward $u$ instead of individual reward $r_t$ as the ranking criterion because the performance of the policy are decided by the entire sequence rather than each state-action pair.
Then, we define a ranking predictor \cite{bradley1952rank} based on the reward function $R_\theta$:
\begin{equation}
    \mathrm{R}(\tau_i\prec\tau_j;\theta)=\frac{\exp\sum_{s_t,a_t\in\tau_j}R_\theta(s_t,a_t)}{\exp\sum_{s_t,a_t\in\tau_i}R_\theta(s_t,a_t)+\exp\sum_{s_t,a_t\in\tau_j}R_\theta(s_t,a_t)}.
\end{equation}
Then, we optimize the network using cross-entropy loss:
\begin{equation}\label{eq:rmloss}
    \mathcal{L}_{RM}=-\mathbb{E}_{(\tau_i,\tau_j,y)\in\mathcal{D}}[(1-y)\log\mathrm{R}(\tau_i\prec\tau_j;\theta)+y\log\mathrm{R}(\tau_j\prec\tau_i;\theta)],
\end{equation}
where $y=\text{int}(\epsilon_i<\epsilon_j)$. 
Intuitively, the cross-entropy loss trains a classifier to predict the quality of two trajectories correctly.

\subsubsection{Inference during reinforcement learning}
Consistent with the ranking criterion in the training process, in subsequent reinforcement learning, the learned reward model $R_\theta$ provides a sparse reward $\hat{r}_t$ for each state-action pair in dance sequence:
\begin{equation}\label{eq:sparsereward}
    \hat{r}_t=\left\{
        \begin{aligned}
            \sum_tr_t&, \text{if}~~~~t=T-1 \\
            0&,\text{else}
        \end{aligned}
    \right.
\end{equation}
We provide a total reward at the end of the dance sequence.

\subsection{Exploration with Reinforcement Learning}\label{sec:rl}

The policy network $\pi_{RL}$ (or $\pi_\phi$) is parameterized by $\phi$ and seeks to maximize the expected return of the trajectory $\tau$:
\begin{equation}
    \begin{aligned}
        \phi&=\argmax_{\phi}\mathbb{E}_{\tau\sim\pi_{\phi}}[R(\tau)]=\argmax_{\phi}\sum_\tau p_{\tau\sim\pi_{\phi}}(\tau)R(\tau),\\
    \end{aligned}
\end{equation}
where $p_{\tau\sim\pi_\phi}(\tau)$ represents the probability of generating trajectory $\tau$ given policy $\pi_\phi$. 
$R(\tau)=\sum_{t=0}^{T-1}\gamma^tR(s_t,a_t)$ represents the discounted return of trajectory $\tau$.
Combined with the Markov Decision Process (MDP) defined in the Preliminary section, for $p_{\tau\sim\pi_\phi}(\tau)$, we have:
\begin{equation}
    \begin{aligned}
p_{\tau\sim\pi_\phi}(\tau)&=p(s_0)\prod_{t=1}^{K-1}\pi_\phi(a_t|s_t)p(s_{t+1}|s_t,a_t) \\
&=p(m_{init},p_{init},m_0)\prod_{t=1}^{K-1}\pi_\phi(p_t|m_{init},p_{init}\cdots,m_t).\\
    \end{aligned}
\end{equation}

The music-conditioned generation is a deterministic environment, where the state and action are given and the next state is deterministic. Therefore, the $p_{\tau\sim\pi_\phi}(\tau)$ is only relative to $\pi_\phi$ and initial state. 
Additionally, since the probability of the initial state $p(s_0)=p(m_{init},p_{init},m_0)$ is determined solely by the environment and not affected by the policy parameters $\phi$, we have:
\begin{equation}
    \begin{aligned}
        \phi&=\argmax_{\phi}\sum_\tau p_{\tau\sim\pi_{\phi}}(\tau)R(\tau)\\
        &=\argmax_{\phi}\sum_\tau\prod_{t=1}^{T-1}\pi_\phi(p_t|m_{init},p_{init},m_0,\cdots,m_t)R(\tau),\\
    \end{aligned}
\end{equation}
where the second term $R(\tau)=\sum_t\gamma^t\hat{r}_t$  is determined by the learned reward model $R_\theta$. 
According to the first term, we can directly apply an auto-regressive model for dance pose generation, \textit{e.g.}, transformer, similar to \cite{chen2021decision, janner2021offline, zheng2022online, xu2022prompting}.
As shown in Figure \ref{fig:RL}, when $\pi_{RL}$
, initialized with $\pi_{BC}$, 
collects samples, the environment provides fixed music sequence $\{m_{init},m_0,\cdots,m_{T-1}\}$ and an initial action $p_{init}$.
The policy $\pi_{RL}$ generates the entire dance pose sequence $\{p_{init},p_0,\cdots,p_{T-1}\}$ in an auto-regressive manner.

\begin{figure}[t]
	\centering
    \includegraphics[width=\linewidth]{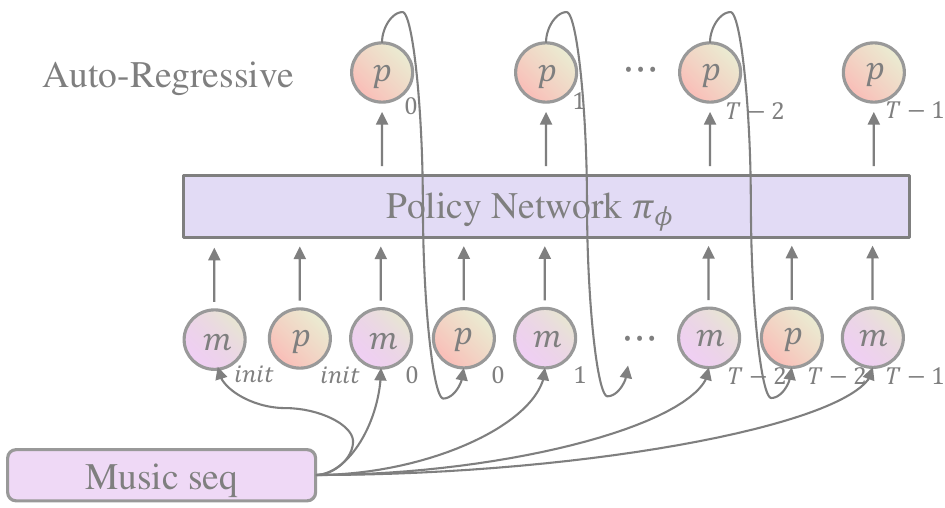}
    \vspace{-5pt}
	\caption{Overview of our policy network. Given a music sequence and an initial pose, the policy network $\pi_{\phi}$ auto-regressively predicts the subsequent poses.}
	\label{fig:RL}
 \vspace{-10pt}
\end{figure}

\begin{table*}[htpb]
 \centering

 \setlength{\tabcolsep}{8pt}
 \begin{tabular}{cccccc}\toprule
    & \multicolumn{2}{c}{Motion Quality} & \multicolumn{2}{c}{Motion Diversity}  \\
    \cmidrule(lr){2-3} \cmidrule(lr){4-5}
    & $FID_{k}\downarrow$  & $FID_{g}\downarrow$ & $DIV_{k}\uparrow$ & $DIV_{g}\uparrow$ & $BAS\uparrow$\\\midrule
    Ground-Truth                & 17.10 & 10.60 & 8.19 & 7.45 & 0.2484  \\\midrule
    FACT \cite{li2021ai}                        & 37.31 & 34.87 & 5.75 & 5.47 & 0.2175  \\
    Bailando \cite{siyao2022bailando}               & 28.62 & 9.95 & 6.27 & 6.22 & 0.2220 \\
    \ModelName~(Ours)                   & \textbf{26.25} &  \textbf{8.94} & \textbf{7.96} & \textbf{6.49} & \textbf{0.2232} \\\bottomrule
 \end{tabular}
    \caption{Evaluation results on test set of different dance generation frameworks. 
 To ensure a fair comparison with baselines, we report the results of \cite{siyao2022bailando} without RL fine-tuning on the test set. 
 }
 \vspace{-10pt}
 \label{tab:experiments}
\end{table*}

The training of the policy adopts the Proximal Policy Optimization algorithm (PPO) \cite{schulman2017proximal}. 
More details are provided in Appendices.

\section{Experiements}


\subsection{Experiments Setup}

\subsubsection{Dataset}

We conduct the training and experiments on the AIST++ dataset \cite{li2021ai}, which is the largest public available dataset for aligned 3D dance motions and music.
AIST++ dataset includes 992 60-Frame Per Second (FPS) 3D dance motion sequences in SMPL \cite{SMPL:2015} format. 
In line with \cite{li2021ai, siyao2022bailando}, we split these data into 952 sequences for training and 40 sequences for subsequent experiments.

\subsubsection{Implementation Details}

For audio preprocess, we employ Librosa to extract music features.
Specifically, we extract the following features: Mel Frequency Cepstral Coefficients (MFCC), MFCC delta, constant-Q chromagram, tempogram, and onset strength, yielding a 438-dimensional musical feature vector.
More details and hyper-parameter settings can be found in Appendices.
 
\subsection{Comparisons with State-Of-The-Arts}

We compare \ModelName~to state-of-the-art including FACT~\cite{li2021ai} and Bailando~\cite{siyao2022bailando}, which is also our behavior cloning policy. 
Following \cite{siyao2022bailando}, we generate 40 dance clips for each method in the AIST++ test set and cut the generated dances into the length of 20 seconds for further experiments.

We conduct objective evaluations following \cite{siyao2022bailando}, including the quality and diversity of generated dances, and the alignment score between the dance and music beats.
Specifically, for the quality of generated dances, we calculate the Fréchet Inception Distance ($FID$) \cite{heusel2017gans} between the generated dance and all dance sequences of AIST++ dataset on the kinetic feature ($FID_{k}$) \cite{onuma2008fmdistance} and the geometric feature ($FID_{g}$) \cite{muller2005efficient}. 
For the diversity, we calculate the average Euclidean distance ($DIV$) \cite{li2021ai} of the generated dances on the kinetic feature ($DIV_{k}$) and the geometric feature ($DIV_{g}$). 
For the alignment between the dance and music beats, we calculate the Beat Align Score ($BAS$) \cite{liu2022learning, siyao2022bailando}.


Table \ref{tab:experiments} report the comparison with state-of-the-art methods.
According to the comparison, the proposed \ModelName~outperforms baseline frameworks in all aspects, demonstrating the effectiveness of the exploration.
Specifically, with exploration, \ModelName~improves 8.28\% and 10.15\% than the BC policy Bailando on $FID_k$ and $FID_g$, respectively.
This indicates that the reward model prefers movements that are more similar to those of humans and high-quality.
And for the motion diversity, exploration helps the policy improve 21.23\% and 4.16\% on $DIV_k$ and $DIV_g$, respectively.
The results on $BAS$ also indicate the improvement of our method.
More comparisons and visualizations in wild musics are available in demo page: {https://sites.google.com/view/e3d2}.
\begin{figure*}[t]
\centering
  \includegraphics[width=0.7\textwidth]{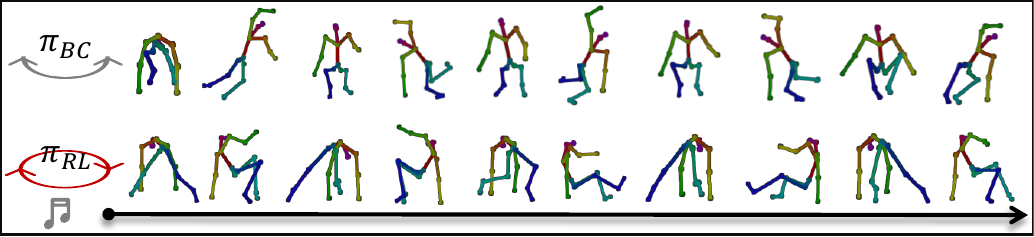}
  \vspace{-5pt}
  \caption{Typical examples. \textit{Top}: 
  Dance movements from the behavior cloning policy; \textit{Bottom}: Dance movements from reinforcement learning policy. 
  Through the exploration, the RL dance agent can make more complex and challenging movements aligned with human preferences, such as "The Thomas Flair", a well-known and demanding movement in street dance.}
  \vspace{-10pt}
  \label{fig:exploration}
\end{figure*}

\section{Discussion}


This section provides a comprehensive analysis of the reward model, including the effectiveness, advantages over the hand-designed reward, the empirical soundness of the training process, and accuracy and generalization.

\subsection{Does exploration provide more diversity and alignment?}\label{sec:rmrole}


As shown in Table \ref{tab:experiments}, the dance generated by \ModelName~ achieves the highest values for both Divergence ($DIV_g$ and $DIV_k$), indicating the diversity of dance poses from \ModelName.
Furthermore, to present the alignment with human preference brought by exploration, we visualize a typical example as demonstrated in Figure \ref{fig:exploration} (more visual comparisons are in the demo page).
The dance sequences above are generated by a BC agent, while the sequences below are generated by an agent finetuned by RL.
Thanks to the exploration, the finetuned dance agent generates complicated and challenging movements aligned with human preferences, such as the \textit{The Thomas Flair}, a challenging move in street dance.
This action is present in the human dataset, but only the RL agent rather than BC agents, did perform it.
It is worth noting that during the RL finetuning, the agent did not have access to any human dataset.
We believe that this phenomenon occurs because, the objective of the behavior cloning policy, predicting the next pose code of dance sequences in the training set, is different from the objective of humans preferring dance sequences.
The frequencies of occurrence of different movements in inference are consistent with their frequencies in the training set, rather than aligned with human preferences, leading to the ordinary movements in the inference process (as \textit{Top} movements in this example).
Later on, during the process of RL, the policy is optimized to act in accordance with human preference through the guidance of the reward model.

\subsection{Is a learned reward function more effective than a hand-designed one?}\label{sec:hdreward}

To explore the differences between learned and hand-designed reward, we compare a hand-designed reward proposed in \cite{siyao2022bailando}:
\begin{equation}
    r = r_b+\gamma_cr_c,
\end{equation}
where $r_b$ and $r_c$ represent the beat alignment reward and orientation reward, respectively. The former aligns dance movements with the rhythm of music, while the latter constrains the consistency of the upper and lower half body.
$\gamma_c$ is the balance weight.

Table \ref{tab:hdreward} presents various evaluation results of dance generation from the agent trained with the hand-designed reward. 
As the interaction steps increase, the dance gradually deviates from human movement patterns and its diversity is not as good as the behavior cloning policy. 
This is because the hand-designed reward only considers the beat alignment and consistency of upper and lower body movements, which results in a lack of diversity and similarity to human movements being overlooked.
Besides, the design of reward requires a lot of task-specific prior knowledge \cite{wirth2017survey,liu2022metarewardnet}.
In contrast, our reward is learned from automatically ranked demonstrations without any domain knowledge. 
This reward is expected to implicitly learn various aspects of the dance, allowing a comprehensive evaluation of the dance and a correct optimization of the dance policy. 

\begin{table}[h]
\begin{center}
\scalebox{0.85}{
\begin{tabular}{cccccc}
\toprule
Steps & $FID_k\downarrow$ & $FID_g\downarrow$ & $DIV_k\uparrow$ & $DIV_g\uparrow$ & $BAS\uparrow$  \\ \hline
0M    & 28.62   & 9.95    & 6.27    & 6.22    & 0.2220 \\
1M    & 45.39   & 15.41   & 4.17    & 3.49    & 0.2338 \\
2M    & 46.25   & 17.20   & 4.63    & 3.46    & 0.2374 \\
3M    & 43.10   & 18.59   & 4.82    & 2.90    & 0.2283 \\
4M    & 47.80   & 22.15   & 4.97    & 2.47    & 0.2388 \\
5M    & 56.30   & 24.58   & 5.52    & 3.56    & 0.2442 \\ \bottomrule
\end{tabular}}
\caption{
Performance of hand-designed reward. `Steps' is the interaction numbers between the agent and the environment. The hand-designed reward only considers $BAS$ (the effect of orientation reward is not included), leading to decreasing performance on other metrics during the optimization.
}
\label{tab:hdreward}
\end{center}
 \vspace{-10pt}
\end{table}

\begin{table}[]
\centering
\scalebox{0.8}{\begin{tabular}{ccccccc}
\toprule
$\epsilon$ & $FID_k\downarrow$ & $FID_g\downarrow$ & $DIV_k\uparrow$ & $DIV_g\uparrow$ & $BAS\uparrow$ & $\overline{u}$  \\ \hline
0.02       & 13.94   & 2.71    & 8.01    & 6.20    & 0.2782 
 & 206.31\\
0.25       & 40.45   & 22.39   & 4.41    & 2.40    & 0.2501 
 & 127.68\\
0.50       & 48.59   & 29.80   & 3.72    & 1.61    & 0.2547 
 & 52.09\\
0.75       & 53.79   & 33.35   & 3.31    & 1.32    & 0.2451 
 & -20.24\\
1.00       & 57.18   & 35.67   & 3.04    & 1.17    & 0.2427 
 & -91.53\\ \bottomrule
\end{tabular}}
\caption{
Ablation on the impact of noise in the training set. The performance of the BC policy gradually decreases as the noise level increases. $\overline{u}$ represents the average total reward across all trajectories in the training set.}
\label{tab:noiseresult}
\end{table}

\subsection{Higher level noise leads to the worse demonstrations?}\label{sec:noiseintensity}

The learned reward model is based on the assumption that the behavior cloning policy significantly outperforms a completely random policy and that increasing levels of noise lead to an increasingly worse policy.
To validate this assumption, we select different noisy levels $[0.02, 0.25, 0.50, 0.75, 1.00]$ and evaluate the performance of behavior cloning policies injected with these noises on the training set music.
As shown in Table \ref{tab:noiseresult}, as the noise level increases, the quality of the generated dances generally decreases.


\begin{figure}[t]
	\centering
    \includegraphics[width=0.75\linewidth]{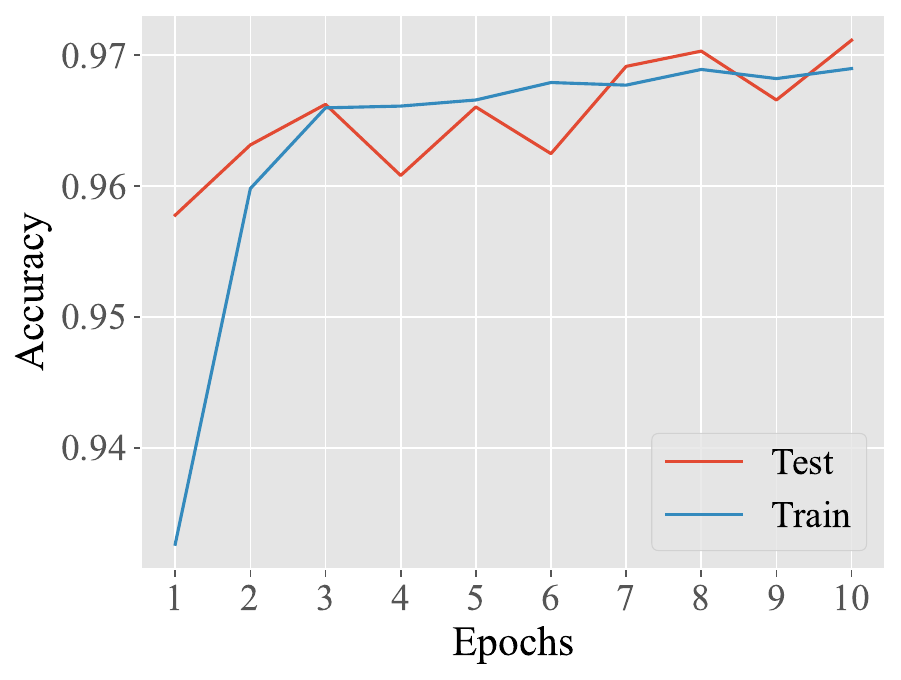}
     \vspace{-5pt}
	\caption{Reward model accuracy: The classification accuracy of the reward model on dances generated by policies with varying levels of noise during training. The reward model exhibits excellent generalization on the test set.}
 \vspace{-10pt}
	\label{fig:rmacc}
\end{figure}

\subsection{What is the performance of the reward model?}\label{sec:rmperformance}

\begin{table}[]
\centering
\scalebox{0.9}{
\begin{tabular}{ccc}
\toprule
Dataset      & Complete Pose & Partial Pose \\ \hline
Music Seen   & 54.69\%       & 73.44\%      \\
Music Unseen & 2.32\%        & 7.52\%       \\ \bottomrule
\end{tabular}}
\caption{Pose prediction accuracy.
We evaluate the behavior cloning policy on both seen and unseen music.
`Complete Pose': both the codes of upper and lower half bodies are correct; `Partial Pose': at least one code is correct.
These results demonstrate the limited generalization capabilities of supervised learning approaches.
}
\label{tab:bcaccuracy}
\end{table}

\begin{table}[]
\centering
\scalebox{0.85}{
\begin{tabular}{cccccc}
\toprule
Dataset & $FID_k\downarrow$ & $FID_g\downarrow$ & $DIV_k\uparrow$ & $DIV_g\uparrow$ & $BAS\uparrow$   \\ \hline
Music Seen       & 8.48   & 1.88    & 8.28    & 6.86    & 0.2854 \\
Music Unseen       & 28.62   & 9.95   & 6.27    & 6.22    & 0.2220  
\\ \bottomrule
\end{tabular}}
\caption{
Performance of behavior cloning policy on seen and unseen music. The significant gap indicates the limited generalization of supervised learning approaches.
}
\vspace{-10pt}
\label{tab:musicseenunseenbc}
\end{table}

Table \ref{tab:noiseresult} demonstrates that the reward model can provide appropriate evaluations, \textit{i.e.}, total reward, for the dance sequences.
For generalization, Figure \ref{fig:rmacc} illustrates the classification accuracy of the reward model on dances with different rankings conditioned on training and test set musics during the training process.
At the final epoch, the reward model achieves an accuracy of around 97\% on both the training and test sets.
In contrast, we evaluate the pose code prediction accuracy of the BC policy on training and test sets, with the results shown in Table \ref{tab:bcaccuracy}.
`Complete Pose' refers to the correct prediction of both the upper and lower half body codes, while `Partial Pose' indicates that at least one code is correct.
The BC policy suffers from severe overfitting to the training set and struggles to generalize to unseen music, achieving only 2.32\% complete pose accuracy and 7.52\% partial pose accuracy.
Moreover, we assess the performance of the BC policy in generating dance sequences under the conditions of seen and unseen music, as shown in Table \ref{tab:musicseenunseenbc}.
The significant gap indicates the weak generalization of the supervised learning approaches.
In comparison, the reward model exhibits excellent generalization performance.
A plausible explanation for these findings is that the reward model, as a discriminative model, easily achieves better generalization than the generative behavior cloning model, enabling reinforcement learning policy to demonstrate superior generalization compared to behavior cloning policy.

\section{Conclusion}
In this paper, to address the problem of the lack of exploration ability in current music-driven dance models, we propose a novel dance generation framework, E3D2. 
We first train a reward model on automatically ranked dance demonstrations, and then, we train the dance policy using reinforcement learning with the learned reward model, resulting in more diverse and human-aligned dances.
Extensive experiments demonstrate the effectiveness of \ModelName.
\clearpage

\bibliography{aaai24}

\onecolumn
\appendix

\section{Theory and Proof}

\subsection{Theorem 1}
We consider that the learned reward model is composed of a feature extraction function $\zeta(\cdot,\cdot)$ and a linear transformation $w$, such as a neural network, where $R(s_t,a_t)=w^\mathrm{T}\zeta(s_t,a_t)$, the learning objective of reinforcement learning, given $R(s_t,a_t)$, can be expressed as the expected discounted return:

\begin{equation}
    J(\pi|R)=\mathbb{E}_{\pi}\left[\sum_{t=0}^{T-1}\gamma^tR(s_t,a_t)\right]=w^\mathrm{T}\mathbb{E}_{\pi}\left[\sum_{t=0}^{T-1}\gamma^t\zeta(s_t,a_t)\right]=w^\mathrm{T}\bm{\zeta}_\pi,
\end{equation}
where $\bm{\zeta}_\pi$ represents the expected discounted features given policy $\pi$.

\begin{theorem}
Let the learned PIRL reward function be $R_\theta(s_t,a_t) = w^\mathrm{T} \zeta(s_t,a_t)$ and the true reward function be $R^*(s_t,a_t) = R_\theta(s_t,a_t) + e(s_t,a_t)$ for some error function $e : \mathcal{S} \rightarrow \mathbb{R}$, and $\|w\|_1 \leq 1$.
Then extrapolation beyond the behavior cloning policy $\pi_{BC}$, i.e., $J(\hat{\pi}|R^*) > J(\pi_{BC}| R^*)$, is guaranteed if : 
\begin{equation}
J(\pi^*_{R^*}|R^*) - J(\pi_{BC}|R^*) > \varepsilon_{\bm{\zeta}}+\frac{2 \|e\|_\infty}{1 - \gamma},
\end{equation}
where $\pi^*_{R^*}$ is the optimal policy under optimal reward function $R^*$, $\varepsilon_{\bm{\zeta}} = \|\bm{\zeta}_{\pi^*_{R^*}} - \bm{\zeta}_{\hat{\pi}}\|_\infty$ and $\|e\|_\infty=\sup\left\{\,\left|e(s_t,a_t)\right|:s_t\in \mathcal{S},a_t\in\mathcal{A}\,\right\}$.
\end{theorem}

\begin{proof}
To prove that the performance of $\pi_{RL}$ surpasses that of $\pi_{BC}$, i.e., $J(\pi_{RL}|R^*)>J(\pi_{BC}|R^*)$, we can compare the performance gap between these two policies and the optimal policy. 
Let $\varrho = J(\pi^*_{R^*}|R^*)-J(\pi_{BC}|R^*)$ as the optimal gap between the optimal policy and behavior policy, if we can prove that $J(\pi^*_{R^*}|R^*)-J(\pi_{RL}|R^*)<\varrho$, then we have $J(\pi_{RL}|R^*)>J(\pi_{BC}|R^*)$.
\begin{equation}
\begin{aligned}
J(\pi^*_{R^*}|R^*)-J(\pi_{RL}|R^*)&=\mathbb{E}_{\pi^*}[\sum_{t=0}^{T-1}\gamma^tR^*(s_t,a_t)]-\mathbb{E}_{\pi_{RL}}[\sum_{t=0}^{T-1}\gamma^tR^*(s_t,a_t)]\\
&=\mathbb{E}_{\pi^*}[\sum_{t=0}^{T-1}\gamma^t(R_\theta(s_t,a_t)+e(s_t,a_t))]-\mathbb{E}_{\pi_{RL}}[\sum_{t=0}^{T-1}\gamma^t(R_\theta(s_t,a_t)+e(s_t,a_t))]\\
&=\mathbb{E}_{\pi^*}[\sum_{t=0}^{T-1}\gamma^t(w^\mathrm{T}\zeta(s_t,a_t)+e(s_t,a_t)]-\mathbb{E}_{\pi_{RL}}[\sum_{t=0}^{T-1}\gamma^t(w^\mathrm{T}\zeta(s_t,a_t)+e(s_t,a_t))]\\
&=w^\mathrm{T}\mathbb{E}_{\pi^*}[\sum_{t=0}^{T-1}\gamma^t\zeta(s_t,a_t)]+\mathbb{E}_{\pi^*}[\sum_{t=0}^{T-1}\gamma^te(s_t,a_t)]-w^T\mathbb{E}_{\pi_{RL}}[\sum_{t=0}^{T-1}\gamma^t\zeta(s_t,a_t)]-\mathbb{E}_{\pi_{RL}}[\sum_{t=0}^{T-1}\gamma^te(s_t,a_t)]\\
&=w^\mathrm{T}\bm{\zeta}_{\pi^*}+\mathbb{E}_{\pi^*}[\sum_{t=0}^{T-1}\gamma^te(s_t,a_t)]-w^T\bm{\zeta}_{\pi_{RL}}-\mathbb{E}_{\pi_{RL}}[\sum_{t=0}^{T-1}\gamma^te(s_t,a_t)]\\
&=w^\mathrm{T}(\bm{\zeta}_{\pi^*}-\bm{\zeta}_{\pi_{RL}})+\mathbb{E}_{\pi^*}[\sum_{t=0}^{T-1}\gamma^te(s_t,a_t)]-\mathbb{E}_{\pi_{RL}}[\sum_{t=0}^{T-1}\gamma^te(s_t,a_t)]\\
&\leq w^\mathrm{T}(\bm{\zeta}_{\pi^*}-\bm{\zeta}_{\pi_{RL}})+(\sup_{s_t\in\mathcal{S},a_t\in\mathcal{A}} e(s_t,a_t)-\inf_{s_t\in\mathcal{S},a_t\in\mathcal{A}} e(s_t,a_t))\sum_{t=0}^{T-1}\gamma^t\\
&=w^\mathrm{T}(\bm{\zeta}_{\pi^*}-\bm{\zeta}_{\pi_{RL}})+(\sup_{s_t\in\mathcal{S},a_t\in\mathcal{A}} e(s_t,a_t)-\inf_{s_t\in\mathcal{S},a_t\in\mathcal{A}} e(s_t,a_t))\frac{1-\gamma^T}{1-\gamma}\\
&\leq w^\mathrm{T}(\bm{\zeta}_{\pi^*}-\bm{\zeta}_{\pi_{RL}})+\frac{(\sup_{s_t\in\mathcal{S},a_t\in\mathcal{A}} e(s_t,a_t)-\inf_{s_t\in\mathcal{S},a_t\in\mathcal{A}} e(s_t,a_t))}{1-\gamma}\\
&\leq w^\mathrm{T}(\bm{\zeta}_{\pi^*}-\bm{\zeta}_{\pi_{RL}})+\frac{2\|e\|_{\infty}}{1-\gamma}\\
&\leq \|w\|_1\|\bm{\zeta}_{\pi^*}-\bm{\zeta}_{\pi_{RL}}\|_{\infty}+\frac{2\|e\|_{\infty}}{1-\gamma}    \qquad \text{H{\"o}lder's inequality}\hfill\\
&\leq \varepsilon_{\bm{\zeta}}+\frac{2\|e\|_{\infty}}{1-\gamma}.
\end{aligned}
\end{equation}

So if $J(\pi^*_{R^*}|R^*)-J(\pi_{BC}|R^*) = \varrho > \varepsilon_{\bm{\zeta}}+\frac{2\|e\|_{\infty}}{1-\gamma}$, then $J(\pi_{R^*}|R^*)-J(\pi_{BC}|R^*) > J(\pi^*_{R^*}|R^*)-J(\pi_{RL}|R^*)$, thus, we have $J(\pi_{RL}|R^*) > J(\pi_{BC}|R^*)$.

\end{proof}

Our proof process draws reference from \cite{brown2020better}.

\section{Visualization of Markov Decision Process}

Figure \ref{fig:MDP} illustrates the Markov Decision Process (MDP) defined in Sec. \ref{sec:preliminary} under the music-conditioned 3D dance generation environment.
The state space $\mathcal{S}$ with $s_t\in \mathcal{P}^{t+1}\times\mathcal{M}^{t+2}$, which means the Cartesian product of $t+1$ dance motion spaces and $t+2$ music feature spaces.
And the action space $\mathcal{A}=\mathcal{P}$, which is discretized with the VQ-VAE, following \cite{siyao2022bailando}.

\begin{figure}[h]
    \centering
    \includegraphics[width=0.6\linewidth]{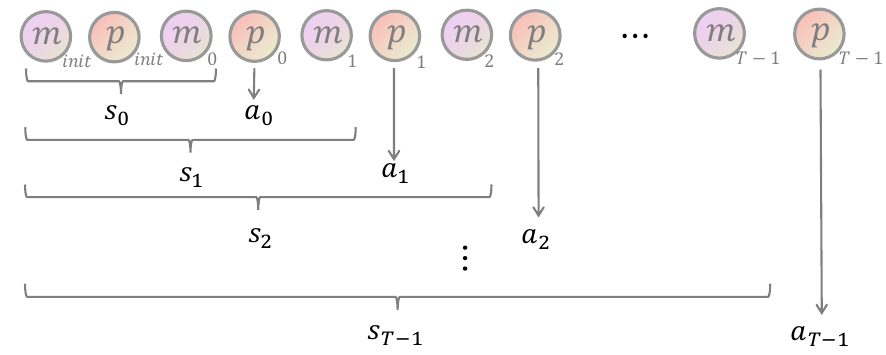}
    \caption{The visualization of the Markov Decision Process (MDP).}
    \label{fig:MDP}
\end{figure}

\section{Complete Algorithm}

Algorithm \ref{alg:e3d2} and Algorithm \ref{alg:rl} showcase the complete algorithmic flow of our proposed \ModelName~framework and the elaborated exploration with reinforcement leanring, respectively.

\begin{figure}[h] 
\begin{minipage}{\linewidth}
\begin{algorithm}[H]
{
\caption{E3D2: Exploratory 3D Dance generation framework}
\label{alg:e3d2}
\begin{algorithmic}[1]
    \State \textbf{Input:} Expert dancer dataset $\mathcal{D}_{human}$, noise schedule $\mathcal{E}$, learning rate $\alpha_0$
    \State Initialize behavior cloning policy $\pi_{\psi}$, reward model network $R_\theta$
    \State \textbf{Output:} Reinforcement learning policy $\pi_{\phi}$
    \State Train the behavior cloning policy $\pi_{BC}$ with $\mathcal{D}_{human}$ in supervised learning\Comment{Sec.~\ref{sec:bc}}
    \For{$\epsilon_i \in \mathcal{E}$}\Comment{Sec.~\ref{sec:collect}}
        \State Run behavior policy with noise $\pi_\psi(\cdot|\epsilon_i)$ in the environment
    \EndFor
    \State Generate automatically-ranked demonstrations dataset $\mathcal{D}_{ranked}$\Comment{Eq.~\ref{eq:dranked}}
    \For{$l=1$ to $L$}\Comment{Sec.~\ref{sec:rewardmodel}}
        \State sample $\tau_i$, $\tau_j$ from $\mathcal{D}_{ranked}$
        \State Update the reward model:
        $$\theta \leftarrow \theta - \alpha_0 \nabla_\theta\mathcal{L}_{RM}$$\Comment{Eq.~\ref{eq:rmloss}}
    \EndFor
    \State Run the Algorithm \ref{alg:rl} with the reward model $R_\theta$ to obtain reinforcement learning policy $\pi_{\phi}$\Comment{Sec.~\ref{sec:rl}}
\end{algorithmic}
}
\end{algorithm}
\end{minipage}
\end{figure}

\begin{figure}[t!] 
\vskip -0.1in
\begin{minipage}{\linewidth}
\begin{algorithm}[H]
{
\caption{Exploration with Reinforcement Learning}
\label{alg:rl}
\begin{algorithmic}[1]
    \State \textbf{Input:}  behavior cloning policy $\pi_\psi$, reward model $R_\theta$, and Dataset $\mathcal{D}=\{(\bm{m}^i,p_{init}^i)\}_{i=1}^M$ of size $M$
    \State Initialize target policy $\pi_{\phi}$, value network $V_{\delta}$,  learning rate $\alpha_{1}, \alpha_{2}$, discount factor $\gamma$
    \For{$n=0$ to $N$} 
        \State Sample  state $(\bm{m},p_{init})$ from $\mathcal{D}$
        \For{$t=0$ to $T-1$}
            \State Run target policy $\pi_{\phi}$ in environment
        \EndFor
        \State Compute sparse rewards $\mathscr{r}$ with KL penalty
        \State Compute advantage estimates $A_0,...,A_{T-1}$
        \State Update the target policy:
        $$\phi \leftarrow \phi-\alpha_1\nabla_\phi\mathcal{L}_{actor}$$
        \State Update the value function:
        $$\delta \leftarrow \delta-\alpha_2\nabla_\delta\mathcal{L}_{critic}$$
    \EndFor
\end{algorithmic}
}
\end{algorithm}
\vspace{-4ex}
\end{minipage}
\end{figure}

\section{Implementation details}
The experiments were run on one NVIDIA A100 GPU.

\subsection{Behavior Cloning}
To ensure a fair comparison, we employed the same hyper-parameters and other implementation details for the behavior cloning policy as those used in Bailando \cite{siyao2022bailando} and specifically referred to the official open-source repository\footnote{https://github.com/lisiyao21/Bailando}.

\subsection{Automatically-Ranked Demonstrations Collection}


After obtaining the behavior cloning policy, $\pi_{BC}$, we collect demonstrations by employing the $\epsilon$-greedy strategy to inject noise into the policy.
At each decision-making step, the agent has a probability of $\epsilon$ to uniformly sample an action $a$ from the action space $\mathcal{A}$, and a probability of $1-\epsilon$ to decide on the action $a$ based on its learned policy $\pi_{BC}$.

In our experiments, we utilized a noise schedule of $\mathcal{E}=\{0.02, 0.25, 0.50, 0.75, 1.00\}$. 
When $\epsilon=1.00$, the policy $\pi(\cdot|\epsilon=1.00)$ becomes entirely random. 
For every piece of music in the $\mathcal{D}_{human}$ dataset, the generated policies interacted with the environment, which resulted in the acquisition of automatically ranked demonstrations across five distinct levels.
To generate our dataset $\mathcal{D}_{ranked}$, we randomly selected two demonstrations with differing levels (potentially driven by different pieces of music) and cropped them into snippets of length 2$\times$30, with 30 tokens for each modality as inputs to the reward model. 
The ground truth for the reward model was determined by the noise level of the demonstrations.
We created 30,000 snippets for both the training and test sets to form the dataset $\mathcal{D}_{ranked}$.


\subsection{Reward Model}
For each demonstration snippet, we remove the $m_{init}$ $p_{init}$ as they are not decided by the policy, forming the snippet with length 2$\times$29, equaling to the context length of the policy.
The specific hyper-parameter settings are shown in the Table \ref{tab:hprm}.

\begin{table}[h]
\centering
\caption{Hyper-parameters of the reward model.}
\begin{tabularx}{0.6\textwidth}{C|C}
\toprule
\textbf{Hyper-parameter}   & \textbf{Value}                              \\ \hline
Hidden Dim        & 768                                \\
Attention Heads   & 12                                 \\
Context Length    & 2 $\times$ 29                        \\
Action Dim        & 512                                \\
Music Dim         & 438                                \\
Block Num         & 1                                  \\
Layer Norm First  & False                              \\
Causal            & True                               \\
Attention Dropout & 0.1                                \\
Residual Dropout  & 0.1                                \\
Optimizer         & Adam                               \\
Learning Rate     & 5e-4                               \\
Weight Decay      & 2.5e-3                             \\
Epoch Num         & 10                                 \\
Batch Size        & 32                                \\
Clip Grad Norm    & 0.25                               \\
Seed              & 0                                  \\ \bottomrule
\end{tabularx}
\label{tab:hprm}
\end{table}

\subsection{Exploration with Reinforcement Learning}

The reinforcement learning policy, $\pi_{RL}$, shares the same architecture as the behavior cloning policy, $\pi_{BC}$, and is initialized using the weights of $\pi_{BC}$. 
This approach is common in many deep reinforcement learning algorithms that employ the Actor-Critic framework.
Both the value network and policy network share an encoder, specifically the first six layers of the Transformer block within the policy network. 
This shared representation helps improve the efficiency and consistency of the learning process.
Detailed hyper-parameter settings for this architecture can be found in Table \ref{tab:hprl}, where the capacity of the buffer is equal to the number of music segments in the training set, as the fact that PPO is an on-policy algorithm and cannot utilize excessively old experiences.


\begin{table}[h]
\centering
\caption{Hyper-parameters of the reinforcement learning.}
\label{tab:hprl}
\begin{tabularx}{0.6\textwidth}{C|C|C}
\toprule
                               & Hyper-parameter     & Value \\ \hline
\multirow{6}{*}{Value Network} & Context Length      & 29    \\
                               & Attention Heads     & 12    \\
                               & Hidden Dim          & 768   \\
                               & Music Emb Dim       & 768   \\
                               & Action Emb Dim      & 768   \\
                               & Block Num           & 3     \\ \hline
\multirow{5}{*}{Buffer}        & Capacity            & 34319 \\
                               & Trajectory Length   & 29    \\
                               & GAE Lambda          & 0.95  \\
                               & Gamma               & 1.0    \\
                               & KL Penalty Beta     & 5e-3  \\ \hline
\multirow{6}{*}{Enviroment}    & Reward Function     & Reward Model \\
                               & Music Normalization & False \\
                               & Train Shuffle       & True \\
                               & Test Shuffle        & False \\
                               & Train Batch Size    & 512 \\
                               & Test Batch Size     & 1 \\ \hline
\multirow{7}{*}{Optimization}         & Train Iters         & 1     \\
                               & Entropy Loss Weight & 0.0   \\
                               & Value Loss Weight   & 0.1   \\
                               & Clip Ratio          & 0.2   \\
                               & Max Grad Norm       & 0.5   \\
                               & Optimizer           & AdamW \\
                               & Learning Rate       & 3e-4  \\ \hline
                               & Seed                & 0     \\ \bottomrule
\end{tabularx}
\end{table}

\section{Background of Reinforcement Learning}

\subsection{From RL to PIRL}
Reinforcement learning aims to learn a mapping from states to actions to maximize a reward signal \cite{sutton2018reinforcement}. 
However, this assumes that the reward function is available, which severely limits the applicability of reinforcement learning.
To address this issue, inverse reinforcement learning (IRL) was proposed \cite{ng2000algorithms} as the inverse problem of reinforcement learning. 
Its goal is to learn a reward function from expert demonstrations, \textit{i.e.}, sequences of state-action pairs, to explain the expert's behavior \cite{ng2000algorithms,abbeel2004apprenticeship}.
Nevertheless, IRL assumes that all examples are provided by the expert and that their appropriateness is consistent.
To address this limitation, \cite{cheng2011preference} introduced preference-based learning to the reinforcement learning field.
Then, \cite{sugiyama2012preference} utilizes preferences to estimate the reward function for dialogue control from rated dialogue sequences and formally proposed Preference-based Inverse Reinforcement Learning (PIRL).
The goal of PIRL is to learn a reward function from expert preferences. 
Unlike IRL, PIRL only requires the expert to provide preference comparisons, making it widely applicable to many tasks.
Recently, the achievement of introducing human feedback (RLHF) into Large Language Model (LLM) \cite{chatgpt,radford2023gpt4} by PIRL has brought PIRL back into the research spotlight.

\subsection{Actor-Critic Framework}
The Actor-Critic framework is a prominent reinforcement learning approach that combines two distinct learning paradigms: policy-based and value-based methods. Within this framework, the actor and critic have distinct roles and responsibilities that contribute to its overall effectiveness.
The actor is in charge of decision-making, determining the optimal course of action given the current state of the environment. Meanwhile, the critic evaluates the value of the selected state or action, providing an estimation of the expected return. The interplay between these two components is critical for refining the learning process.
The actor optimizes its policy based on the chosen action and the value estimation provided by the critic, with the aim of improving the expected discounted return. This optimization is achieved through the use of policy gradient (PG) algorithms. On the other hand, the critic focuses on updating its value estimation function according to the actions taken by the actor and the actual reward signal. This process allows the critic to enhance the accuracy of its value predictions and is typically implemented using temporal difference (TD) learning algorithms.
Through this synergistic relationship between the actor and the critic, the Actor-Critic framework effectively balances exploration and exploitation, resulting in a more robust and efficient reinforcement learning process.

\subsection{Reinforcement Learning for Media Generation}
Recently, ChatGPT \cite{chatgpt} and GPT-4 \cite{radford2023gpt4} have garnered significant interest from both academia and industry. 
Large language models fine-tuned through Reinforcement Learning from Human Feedback (RLHF) \cite{christiano2017deep, warnell2018deep,macglashan2017interactive} achieve impressive performance on multiple natural language generation tasks, such as multi-turn dialogue generation \cite{ouyang2022training, bai2022training, glaese2022improving}, text summarization \cite{stiennon2020learning}, information and QA \cite{menick2022teaching}, \textit{etc}.
Apart from natural language, Reinforcement Learning also exists extensively in other media content generation tasks, including image generation \cite{huang2019learning, sun2022learning}, video summarization \cite{zhou2018deep}, text-to-speech \cite{liu2021reinforcement}, music synthesis \cite{jaques2016generating} and gesture generation \cite{sun2023cospeech}.
Specially, in the field of motion generation, \cite{toverud2022reinforcement} used RL and principal component analysis for music-free dance motion generation with hand-designed reward function. 
Similarly, \cite{siyao2022bailando} used RL to optimize their model during the fine-tuning process on the test set, employing a hand-designed reward function to enhance the alignment of dance movements with the beat and orientation. 
However, hand-designed reward for certain targets may inadvertently impair the agent's performance in other aspects (as described in Section \ref{sec:hdreward}). 
Additionally, the design of reward functions depends heavily on domain experience and requires constant trail and verification, which can be challenging and time-consuming \cite{wirth2017survey,liu2022metarewardnet}.
In contrast, our approach learns the reward function by distinguishing different noise levels injected into demonstrators, offering a simpler, more comprehensive, and more accurate solution. 
Most importantly, we optimize our policy and sufficiently unlock the potential of exploration during the training process. 
This sets our method apart from the \cite{siyao2022bailando}, leading to a robust and adaptable dance generation model.

\section{Introduction to Proximal Policy Optimization and Its Components}

\subsection{Vanilla Policy Gradient}
Vanilla policy gradient methods combine the increase in the probability of selecting high-return actions by the policy with multi-step Markov decision processes. It aims to optimize a policy directly by following the gradient of the expected return with respect to the policy parameters. The gradient of the expected return is expressed as:

\begin{equation}
\nabla_{\phi} \mathbb{E} \left[\sum R_t \right] = \mathbb{E} \left[\sum \left(\nabla_{\phi} \log(\pi_{\phi}(a_t|s_t)) \cdot r(s_t, a_t)\right) \right],
\end{equation}
where $r(s_t, a_t)$ denotes the return for taking action $a_t$ in state $s_t$.

\subsection{Advantage Function}
The advantage function measures the relative benefit of taking an action $a_t$ in state $s_t$ compared to the average action in that state. The advantage function can be defined using the return $r(s_t, a_t)$, the state-value function $V_{\pi}(s_t)$, and the next state's value function $V_{\pi}(s_{t+1})$:

\begin{equation}
A_{\pi}(s_t, a_t) = r(s_t, a_t) + \gamma V_{\pi}(s_{t+1}) -V_{\pi}(s_t),
\end{equation}
where $\gamma$ represents the discount factor. Employing the advantage function $A_{\pi}(s_t, a_t)$ in place of the return $r(s_t, a_t)$ enables the algorithm to concentrate on enhancing the most promising actions while reducing the variance in gradient estimation, ultimately improving learning efficiency.

The formula above represents the calculation of the advantage in temporal difference form. However, it has the disadvantage of introducing bias due to the biased estimate of the return.
In this paper, we adopt a formula for calculating the advantage that directly utilizes the return, which avoids introducing bias and also reduces variance.

\begin{equation}
A_{\pi}(s_t, a_t) = \sum_{t=0}^{T-1}\gamma^tr(s_t,a_t) -V_{\pi}(s_t)
\end{equation}

\subsection{Importance Sampling}
Importance sampling is introduced to account for the discrepancy between the old and new policies when updating the policy parameters. This technique allows for the estimation of the expected return under the new policy ($\pi_{\phi}$) using samples generated from the old policy ($\pi_{\phi'}$). The importance sampling ratio is given by:

\begin{equation}
\rho_t = \frac{\pi_{\phi}(a_t|s_t)}{\pi_{\phi'}(a_t|s_t)}
\end{equation}

\subsection{PPO Clipping}
To maintain stability during the optimization process, Proximal Policy Optimization (PPO) introduces a clipped surrogate objective function. The surrogate objective function is optimized multiple times using the advantage $A_t$ calculated from trajectories collected by the old policy, in order to achieve a monotonic and stable performance improvement of the policy. However, as PPO is an on-policy algorithm that uses the old advantage $A_t$ for multiple optimization steps, it is necessary to constrain the difference between the old and new policies. The rationale for clipping is to limit the policy update, ensuring that the new policy does not deviate excessively from the old policy. The PPO clipped objective function is:

\begin{equation}
L(\phi) = \mathbb{E}_{\tau\sim\pi_{\phi'}}\left[\sum_{t=0}^T\min\left(\rho_t A_t,\text{clip}(\rho_t,1-\eta,1+\eta)A_t)\right)\right],
\end{equation}
where $\eta$ is a small positive constant \textit{e.g.}, 0.1 or 0.2.

In summary, PPO incorporates vanilla policy gradient, advantage function, importance sampling, and a clipped surrogate objective function to create a more stable and efficient optimization process for MDP problems.








\end{document}